\documentclass[a4paper,pra,twocolumn,superscriptaddress,groupedaddress]{revtex4}
\usepackage{amsmath,amscd,amsfonts,amssymb, color,bbm, bm,
	braket, xcolor, setspace, cancel, stmaryrd, calc, amsthm} 
\usepackage{ae}
\usepackage{physics}
\usepackage[T1]{fontenc}
\usepackage[ansinew]{inputenc}
\usepackage{amsmath}
\usepackage{amssymb}
\usepackage[caption=false]{subfig}
\usepackage{multirow}
\usepackage{array}
\usepackage[]{graphicx}
\usepackage{wrapfig}
\usepackage{makecell}
\usepackage{xparse}
\usepackage{dcolumn}
\usepackage{color}
\usepackage[colorlinks=true, linkcolor=red, citecolor=blue, urlcolor=blue]{hyperref}
\usepackage{lscape}
\newtheorem{theorem}{Theorem}
\graphicspath{ {./images1/} }

\newtheorem{prop}[theorem]{Proposition}
\newtheorem{defin}{Definition}
\newtheorem{rem}{Remark}

\begin{document}
	\title{Classifying locally distinguishable sets: No activation
		across bipartitions}
	\author{Atanu Bhunia}
	\email{atanu.bhunia31@gmail.com}
	\affiliation{Department of Mathematical Sciences, Indian
		Institute of Science Education and Research Berhampur, Laudigam,
		Konisi, Berhampur 760 003,
		Odisha, India}
	\author{Saronath Halder}
	\email{saronath.halder@vitap.ac.in}
    \affiliation{Department of Physics, School of Advanced Sciences, VIT-AP University, Beside AP Secretariat, Amaravati 522241, Andhra Pradesh, India}
	\author{Ritabrata Sengupta}
	\email{rb@iiserbpr.ac.in}
	\affiliation{Department of Mathematical Sciences, Indian
		Institute of Science Education and Research Berhampur, Laudigam,
		Konisi, Berhampur 760 003,
		Odisha, India}
	
	\begin{abstract}
		A set of orthogonal quantum states is said to be locally
		indistinguishable if they cannot be perfectly distinguished by local
		operations and classical communication (LOCC). Otherwise, the states
		are locally distinguishable. Interestingly, locally indistinguishable states can have productive applications in quantum information processing protocols. In this sense, locally indistinguishable states are useful. On the other hand, it is usual to consider that locally distinguishable states are
		useless. Nevertheless, recent works suggest that locally
		distinguishable states should be given due consideration as in certain
		situations these states can be converted to locally indistinguishable
		states under orthogonality-preserving LOCC (OP-LOCC). Such a
		counterintuitive phenomenon motivates us to ask when the aforesaid
		conversion is possible and when it is not. In this work, we provide
		different structures of locally distinguishable product and entangled
		states which do not allow the aforesaid conversion. We also provide
		certain structures of locally distinguishable states which allow the
		aforesaid conversion. In this way, we classify the locally
		distinguishable sets by introducing hierarchies among them. In a
		multipartite system, this study becomes more involved as there exist
		multipartite locally distinguishable sets which cannot be converted to
		locally indistinguishable sets by OP-LOCC across any bipartition. We
		say this as ``no activation across bi-partitions".
		
	\end{abstract}
	%\date{\today}
	%\pacs{}
	\maketitle
	
	\section{Introduction}
	\par Non-local properties of quantum systems have a class
	exclusive from Bell non-locality \cite{Bell non-locality}. Specifically,
	when a set of orthogonal quantum states cannot be perfectly
	distinguished by local operations and classical communication (LOCC),
	it reflects a fundamental nonlocal feature of quantum physics
	\cite{BennettPB1999}. Local distinguishability of quantum states
	refers to the task of identifying a state from a set of prespecified
	orthogonal states shared among parties separated by arbitrary
	distances and LOCC being the only legit class of operation
	\cite{BennettUPB1999,Walgate2000,Virmani,Ghosh2001,Groisman,Walgate2002,Divincinzo,Horodecki2003,Fan2004,Ghosh2004,Nathanson2005,Watrous2005,Niset2006,Ye2007,Fan2007,Runyo2007,somsubhro2009,Feng2009,Runyo2010,Yu2012,Yang2013,Zhang2014,somsubhro2010,yu2014,somsubhro2014,somsubhro2016,bennett1996,popescu2001,xin2008,somsubhro2009(1)}.
	The non-locality of orthogonal quantum states can be used for various
	practical purposes such as data hiding
	\cite{terhal58,divincenzo580,lamidatahiding,terhaldatahiding,chaves2020,wehner2020,winterdatahiding,haydendatahiding},
	quantum secret sharing \cite{rahaman330,markham309,wang320}, and
	similar applications. Consequently, in the past two decades,
	considerable attention has been paid to the study of local
	distinguishability of orthogonal quantum states and the exploration of
	the relationship between quantum non-locality and entanglement
	\cite{Zhang2015,Wang2015,Chen2015,Yang2015,Zhang2016,Xu2016(2),Zhang2016(1),Xu2016(1),Halder2019strong
		non-locality,Halder2019peres
		set,Xzhang2017,Xu2017,Wang2017,Cohen2008,somsubhro2018,zhang2018,Halder2018,Yuan2020,Rout2019,bhunia2020,bhunia2023,biswas2023,Zhang2019,bhunia2022}.
	\par In quantum information processing, one of the most
	important physical scenario occurs when a multipartite system is
	distributed to different parties separated by arbitrary distances. The
	parties perform multiple rounds of local measurements on their
	respective subsystems, each time globally broadcasting their
	measurement outcomes. Other parties then choose their measurement
	setups depending on the outcomes and continue till required. This
	class of operations is known as LOCC. From an experimental
	perspective, LOCC operations have a natural attraction since local
	quantum measurements are much easier to perform on a composite system
	than their nonlocal counterparts. In fact, on a more fundamental
	level, LOCC is linked to the very notion of entanglement since
	entanglement is precisely the multipartite correlations that cannot be
	generated by LOCC \cite{quantum entanglement}. However, despite this
	general importance, the class of LOCC is still not satisfactorily
	understood.
	\par Local distinguishability of quantum states plays an important
	role in studying the restrictions of LOCC. In 2000 Walgate et
	al.~\cite{Walgate2000} have evinced that any two orthogonal
	multipartite pure states can be perfectly distinguished by allowing
	LOCC. Nevertheless, if there are more than two orthogonal pure states,
	then there can be local indistinguishability. The local
	indistinguishability of a set of pairwise orthogonal multipartite
	states is a signature of non-locality shown by those states.
	Since entanglement is intrinsically
	connected to non-locality, one can assume that mutually orthogonal
	product states can be perfectly distinguished by LOCC. However,
	entanglement is not necessary for local indistinguishability of
	quantum states
	\cite{BennettPB1999,BennettUPB1999,Zhang2015,Wang2015,Chen2015,Yang2015,Zhang2016,Xu2016(2),Zhang2016(1),Xu2016(1),bhunia2023,
		biswas2023, Halder2019strong non-locality,Halder2019peres
		set,Xzhang2017,Xu2017,Wang2017,Cohen2008,Zhang2019,somsubhro2018,zhang2018,Halder2018,Yuan2020,Rout2019,bhunia2020,bhunia2022,bhuniaubb2024,indra2025,subrata2024,bhunia2025}.
	In 1999 Bennett et al.~\cite{BennettPB1999} first exhibited a set of
	nine product states in a two-qutrit system, which cannot be
	perfectly distinguished by LOCC and presented the phenomenon of
	``non-locality without entanglement''. 
    %The result indicates that
	%entanglement is not a requisite factor of local indistinguishability
	%of quantum states. 
	
%	This manifests that the absence of entanglement is
%	not sufficient to ensure the local accessibility of information.
%	Furthermore, there is an incomplete basis for demonstrating the
%	phenomenon of non-locality without entanglement, commonly known as the
%	unextendible product basis (UPB). It is defined by a set of mutually
%	orthogonal product states satisfying the condition that the orthogonal
%	complement of the subspace, spanned by all these product states,
%	contains no product states, i.e., this set of states cannot be
%	extended to a complete basis by adding product states to it while
%	preserving the orthogonality of the set \cite{Divincinzo}. UPB cannot
%	be accurately distinguished by LOCC \cite{LOCC}, and the normalised
%	projector onto the orthogonal complement of it is a mixed state which
%	uncovers a captivating phenomenon known as bound entanglement
%	\cite{BennettUPB1999,Divincinzo}. Thus, these states are of
%	considerable interest in quantum information theory.
	
	Due to practical applications, local indistinguishability of
	quantum states can be considered as a resource in quantum information
	processing. Now, if there are only locally distinguishable sets at hands,
	how can we transfer them into resources that have applications in data
	hiding? This is what the authors of the paper \cite{bandhyopadhyay201}
	recently studied. In fact, they studied the following problem: is
	there any set of orthogonal states which can be locally
	distinguishable, but under an orthogonality-preserving (OP) local
	measurement, each outcome will lead to a locally indistinguishable
	set. As there are some trivial sets with this property, they
	introduced the concept of local irredundancy. An orthogonal set is
	said to be locally redundant if it remains orthogonal after discarding
	one or more subsystems. Otherwise, it is said to be locally
	irredundant. If a locally irredundant set satisfies the aforementioned
	property, then we say that its non-locality can be activated genuinely,
	i.e., hidden non-locality can be revealed. In
	Ref.~\cite{bandhyopadhyay201}, the authors provided several examples
	of such sets with entanglement. However, deeper research on this
	property remains to be explored.
	For example, the following questions are required to be
	studied. Is there any multipartite locally distinguishable set (with
	or without entanglement) whose non-locality cannot be activated even if
	(specific) joint operations are allowed? In which multipartite state
	spaces can such locally distinguishable sets be constructed? Answering such
	questions are particularly important to understand when one can have
	activation of non-locality. See also \cite{subrata2024} in this regard. In the following, we provide a table (Table~\ref{tab:locc}) where different activation results and corresponding references are summarized. From this table, it can be understood what has been done so far and what are the present contributions.
    \begin{widetext}
    \begin{center}
    \begin{table}[h]
	\centering
	\renewcommand{\arraystretch}{1}
	\setlength{\tabcolsep}{1.5pt}
	\begin{tabular}{|l|l|}
	\hline
	\textbf{Results} & \textbf{Status}\\[1 ex]
    \hline\hline
	Locally distinguishable sets of orthogonal entangled states can & Bipartite and multipartite examples \\
    be converted to locally indistinguishable sets under OP-LOCC. & are constructed in \cite{bandhyopadhyay201}.\\[1 ex]
    \hline
    Similar results as described in the above row are observed & Bipartite and multipartite examples \\   
    considering orthogonal product states. & are constructed in \cite{Li2022}.\\[1 ex]
    \hline
    Locally distinguishable sets of orthogonal product states can & Bipartite examples are not possible.\\
    be converted to locally indistinguishable sets under OP-LOCC, & Thus, multipartite examples are \\
    such that the sets exhibit local indistinguishability across & constructed in \cite{GhoshStrongActivation2022}.\\
    every bipartition. & \\[1 ex]
    \hline
    Locally distinguishable sets of orthogonal states cannot be & Bipartite examples are not possible.\\
    converted to locally indistinguishable sets under OP-LOCC & Thus, multipartite examples are \\
    when all parties are parties are spatially separated. However, & constructed in \cite{subrata2024}.\\
    it can be done when two parties collaborate and perform joint & \\
    measurements. & \\[1 ex]
    \hline
	There are instances when locally distinguishable sets of & Bipartite examples are not possible.\\
    orthogonal states cannot be converted to locally indistinguishable & Thus, multipartite examples are \\
    sets under OP-LOCC even if joint measurements are allowed (all & constructed here. We also develop\\
    but one parties are allowed to collaborate and perform joint & methodology for different constructions.\\
    measurements).                                           & A possible classification among existing\\
    & and presently provided structures are\\
    & discussed. (present contributions)\\[1 ex]
    \hline
    \end{tabular}
	\caption{Summary of LOCC activation results and corresponding references are given.}\label{tab:locc}
	\end{table}
	  \end{center}
    \end{widetext}
   %Notably, certain configurations require joint measurements or coordinated two-way protocols to achieve perfect discrimination, while others remain indistinguishable under all LOCC constraints. Checkmarks (\checkmark) indicate whether activation is possible under LOCC. \textsuperscript{(\#)} Denotes a result derived in this work.
	
	In the process of studying the aforesaid questions, here we
	manage to construct sets of multipartite states which are not activable
	in any bipartition. In other words, such locally distinguishable sets
	cannot be transformed to a locally indistinguishable set in any
	bipartition under orthogonality-preserving LOCC. We say this as the worst-case scenario in view of
	non-locality activation. Because if we consider all subsystems together
	in a single location, then, anyway, there will be no local
	indistinguishability as we are dealing with orthogonal states here.
	The discovery of this class of sets also leads to a hierarchy among
	the multipartite locally distinguishable sets. The structures that we
	provide here can be easily generalised. In particular, for bipartite
	systems, we consider higher-dimensional Hilbert spaces compared to
	some known results of two-qubit or qubit-qudit cases
	\cite{bandhyopadhyay201}. Then, we compare between locally
	distinguishable product states and entangled states. 	
	The paper is organised as follows: in \S\ref {A1},
	necessary definitions and other preliminary concepts are presented. In
	\S\ref{A2}, we provide activable and non-activable sets of product
	states in bipartite as well as in multipartite scenarios. In \S\ref{A3}, we consider entangled states and present comparisons
	between product states and entangled states.     Finally, the
	conclusion is drawn in \S\ref{A4}.
	\section{Preliminaries}
	\label{A1}
	A measurement on a $d$-dimensional quantum system can be
	expressed as a set of positive operator-valued measure (POVM) elements
	$\left\{M_k\right\}_k$. These elements are the positive semidefinite
	Hermitian matrices that satisfy the completeness relation $\sum_k
	M_k=\mathrm{I}_d$, where $\mathrm{I}_d$ is the identity matrix of
	order $d$. In this section, we will first review some of the
	definitions which are used throughout the following sections.
	\begin{defin}
		\cite{Walgate2002, Halder2018} If all the POVM elements of
		a measurement structure, corresponding to a discrimination task of a
		given set of states, are proportional to the identity matrix, then
		such a measurement is not useful to extract information for this task
		and is called a \emph{trivial measurement}. Conversely, there should be at
		least one POVM element not proportional to the identity matrix, then the 
		measurement is called as \emph{non-trivial}.	    
		%	    On the other hand, if not all POVM elements of a measurement
		%	are proportional to the identity matrix, then the measurement is said to be a \emph{nontrivial measurement}.
	\end{defin}
	\begin{defin} \cite{Walgate2002, Halder2018}  Consider a local
		measurement to distinguish a fixed set of pairwise orthogonal quantum
		states. If the post-measurement states likewise exhibit the
		property of pairwise orthogonality, then such a measurement shall be
		termed as an \emph{orthogonality-preserving local measurement (OPLM)}.
	\end{defin}
	\noindent    In this work, we always stick to OPLM. 
	\begin{defin}
		\cite{Halder2019strong non-locality} A set of orthogonal
		quantum states is \emph{locally irreducible} if it is not possible to
		eliminate one or more quantum states from the set by nontrivial
		orthogonality-preserving local measurements.
	\end{defin}
	\begin{defin} 
		A set of orthogonal quantum states is said to be \emph{locally
			indistinguishable} if, whilst it may be possible to eliminate one or
		more states from the set via an OPLM, it is impossible to
		completely distinguish the entire set using a non-trivial OPLM.
	\end{defin}
	\par Therefore, it is by definition implied that all locally
	irreducible states are locally indistinguishable but the converse is
	not true.   
	
	\begin{defin}
		A locally distinguishable set $\mathcal{S}$ of multipartite
		orthogonal states is said to be locally activable if it can be
		transformed to a set of locally indistinguishable orthogonal states
		via local orthogonality-preserving measurements.    
	\end{defin} 
	\par Let us assume that the total number of parties is $N$. 
	\begin{defin}
		A locally distinguishable set of multipartite orthogonal states, S, is
		deemed to possess \emph{hidden non-locality of type-1} if, upon spatial
		separation of all constituent parties, the set may be activated by
		means of LOCC.
		%    A locally distinguishable set of multipartite orthogonal states $\mathcal{S}$ is said to have hidden non-locality of type-1, if it is possible to activate the set by LOCC when all parties are spatially separated. 
		We denote this by $\mathcal{H}^{\text{LOCC}}_{1}(\mathcal{S})\neq0$.
		Also, if a locally distinguishable set of multipartite orthogonal
		states $\mathcal{S}$ is said to have \emph{hidden non-locality of
			type-$k$}, if to activate the set by LOCC, at least $k$ parties are
		needed to come together, whereas all other parties are spatially
		separated. We denote this by
		$\mathcal{H}^{\text{LOCC}}_{k}(\mathcal{S})\neq0$.
	\end{defin}
	
	The maximum value of $k$ in
	$\mathcal{H}^{\text{LOCC}}_{k}(\mathcal{S})$ can be equal to $(N-1)$,
	because, if $k=N$, then, all parties are coming together and there is
	no local indistinguishability. This happens as we are dealing with
	orthogonal states. Naturally, in a bipartite scenario, the only case
	that appears is $k=1$. In this work, when we consider a particular value 
    of $k$, then to check activability/non-activability, we consider local 
    distinguishability/indistinguishability in $k$-partitions only. If for 
    $k=N-1$, it is not possible to convert a locally distinguishable set to 
    an indistinguishable one under OP-LOCC, i.e., when we get non-activability 
    in all $k$-partitions for maximum value of $k$, then we say it as worst 
    case scenario here.
	
	%%%%%%%%%%%%%%%%%%%%%%%%%%%%%%%%%%
	\section{Non-activable and activable product states}\label{A2}
	In this section, we first construct a class of orthogonal
	product states which cannot be activable by LOCC. For better
	understanding, we first give an example in
	${\mathbb{C}}^{4}\otimes{\mathbb{C}}^{4}$ and then,
	generalise the result. 	
	\begin{figure}[h!]
		\centering
		\includegraphics[width=0.36\textwidth]{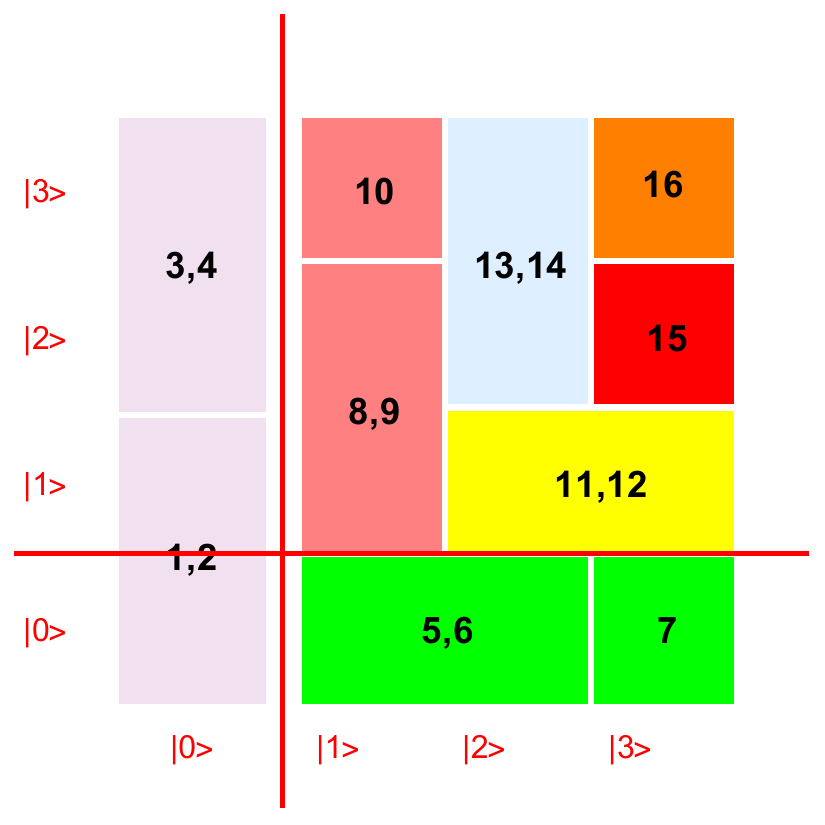}
		\caption{\emph{(Color online)} Representation of product states in ${\mathbb{C}}^{4}\otimes{\mathbb{C}}^{4}$. 
			The bottom side represents Alice's side and
			top left side represents Bob's side (this is also maintained in other
			figures unless explicitly stated). We represent quantum states
			$\mathbf{\ket{i\pm\overline{i+1}
				}\ket{j}}$ or, $\mathbf{\ket{j}\ket{i\pm\overline{i+1}
			}}$ by rectangular tiles where $\mathbf{\ket{i\pm\overline{i+1}
				}=\frac{1}{\sqrt{2}}(\ket{i}\pm\ket{i+1})},$
			for integer `$i$'. Each of the square tiles represents a state of the
			form $\mathbf{\ket{j}\ket{k}}$. Tile indices correspond to consecutive
			ordered basis states of set $\mathcal{S}_{1}$, while tile colours
			indicate compatible measurement setups for both
			parties.}\label{fig44}
	\end{figure}
	
	\par Consider the set $\mathcal{S}_1 \subseteq \mathbb{C}^4 \otimes \mathbb{C}^4$, by,
	\begin{multline}
		\mathcal{S}_1=
		\begin{Bmatrix}
			\mathbf{|0\rangle_{A}|{\mathcal{X}}_{01}^{\pm}\rangle_{B},}\;
			\mathbf{|0\rangle_{A}|{\mathcal{X}}_{23}^{\pm}\rangle_{B},}\;
			\mathbf{|{\mathcal{\xi}}_{12}^{\pm}\rangle_{A}|0\rangle_{B},}\;\\
			\mathbf{|{\mathcal{\xi}}_{3}\rangle_{A}|0\rangle_{A},}\;
			\mathbf{|1\rangle_{A}|{\mathcal{X}}_{12}^{\pm}\rangle_{B},}\;
			\mathbf{|1\rangle_{A}|{\mathcal{X}}_{3}\rangle_{B},}\;\\
			\mathbf{|{\mathcal{\xi}}_{23}^{\pm}\rangle_{A}|1\rangle_{A},}\;
			\mathbf{|2\rangle_{A}|{\mathcal{X}}_{23}^{\pm}\rangle_{B},}\;\\
			\mathbf{|{\mathcal{\xi}}_{3}\rangle_{A}|2\rangle_{B},}\;
			\mathbf{|{\mathcal{\xi}}_{3}\rangle_{A}|{\mathcal{X}}_{3}\rangle_{B}}\\
		\end{Bmatrix}
		\label{1}
	\end{multline}
	where,
	$\mathbf{|{\mathcal{\xi}}_{ij}^{\pm}\rangle_{A}=\left(\frac{|i\rangle\pm
			|j\rangle}{\sqrt{2}}\right)_{A},\;|{\mathcal{X}}_{ij}^{\pm}\rangle_{B}=\left(\frac{|i\rangle\pm
			|j\rangle}{\sqrt{2}}\right)_{B}},$ and
	$\mathbf{|{\mathcal{\xi}}_{k}\rangle_{A}=|k\rangle_{A},\;|{\mathcal{X}}_{k}\rangle_{B}=|k\rangle_{B}}$,
	see Fig.~\ref{fig44}.
	
	\begin{prop}\label{prop1}
		The set $\mathcal{S}_1$ does not possess any activable non-locality
		under orthogonality-preserving LOCC, i.e.,
		$\mathcal{H}^{\text{LOCC}}_{1}(\mathcal{S}_1) = 0$.
	\end{prop}
	\begin{proof} 
		Suppose Alice makes the first measurement. Let 
		${\mathcal{M}_{A}^m}^\dagger\mathcal{M}_{A}^m=[m^a_{ij}]_{4\times4}$ denote any arbitrary POVM operator of Alice with outcome
		$m$ such that the post-measurement  states $\left\{\mathcal{M}_{A}^m \otimes I_{B} \left|\psi_{i}\right\rangle,|\psi_{i}\rangle\in \mathcal{S}_1\right \}$ are mutually
		orthogonal. As, $m^a_{i j}=0$ is necessary and sufficient for
		$m^a_{j i}=0,\, i<j$. We only need to show that $m^a_{i j}=0, i<j$, in the following. 
		For the states $
		\mathbf{|0\rangle_{A}|{\mathcal{X}}_{01}^{+}\rangle_{B}}$ and $\mathbf{
			|1\rangle_{A}|{\mathcal{X}}_{12}^{+}\rangle_{B}},$ it is easy to see that 
		$\left\langle
		\mathbf{0}\left|{\mathcal{M}_{A}^m}^\dagger\mathcal{M}_{A}^m\right|
		\mathbf{1}\right\rangle_{\mathbf{A}} \mathbf{\left\langle 0+1|1+2\right\rangle_{B}}=0.$ Thus,
		$m^a_{01}=m^a_{10}=0 .$ 
		\par In the same way, for the states $\{\mathbf{|0\rangle_{A}|{\mathcal{X}}_{23}^{+}\rangle_{B},
			|2\rangle_{A}|{\mathcal{X}}_{23}^{+}\rangle_{B}}\},$ and $\{\mathbf{|0\rangle_{A}|{\mathcal{X}}_{23}^{+}\rangle_{B},
			|{\mathcal{\xi}}_{3}\rangle_{A}|2\rangle_{B}}\}$,
		we can see that $m^a_{02}=m^a_{20}=0, m^a_{03}=m^a_{30}=0$, respectively. Similarly, if we choose the states $	
		\{\mathbf{|1\rangle_{A}|{\mathcal{X}}_{12}^{+}\rangle_{B},
			|2\rangle_{A}|{\mathcal{X}}_{23}^{+}\rangle_{B}}\},$ and $\{\mathbf{|1\rangle_{A}|{\mathcal{X}}_{12}^{+}\rangle_{B},
			|{\mathcal{\xi}}_{3}\rangle_{A}|2\rangle_{B}}\}$,
		we get $m^a_{12}=m^a_{21}=0, m^a_{13}=m^a_{31}=0$. 
		\par Now considering the states $\{\mathbf{|{\mathcal{\xi}}_{3}\rangle_{A}|2\rangle_{B},
			|2\rangle_{A}|{\mathcal{X}}_{23}^{+}\rangle_{B}}\}$. For this, we have
		$\left\langle
		\mathbf{2}\left|{\mathcal{M}_{A}^m}^\dagger\mathcal{M}_{A}^m\right|\mathbf{3}\right\rangle_{\mathbf{A}}
		\mathbf{\left\langle 2|2+3\right\rangle_{B}}=0$, which implies that
		$m^a_{23}=m^a_{32}=0 .$ Therefore, with respect to the basis $\mathcal{S}_1$,
		${\mathcal{M}_{A}^m}^\dag \mathcal{M}_{A}^m$ is diagonal and
		${\mathcal{M}_{A}^m}^\dagger\mathcal{M}_{A}^m$ can be written as $\mathrm{diag}\left(\delta_{0},
		\delta_{1},\delta_{2},\delta_{3}\right)$.
		\par Now considering
		$\mathbf{|{\mathcal{\xi}}_{12}^{\pm}\rangle_{A}|0\rangle_{A}},$ we get
		$\mathbf{\quad\left\langle1+2\left|{\mathcal{M}_{A}^m}^\dagger\mathcal{M}_{A}^m\right|1-2\right\rangle_{A}\langle
			0|0\rangle_{B}}=0, \quad$ i.e.
		$\mathbf{\left\langle1\left|{\mathcal{M}_{A}^m}^\dagger\mathcal{M}_{A}^m\right|1\right\rangle-\left\langle
			2\left|{\mathcal{M}_{A}^m}^\dagger\mathcal{M}_{A}^m\right|2\right\rangle}=0$.
		Thus, $m^a_{11}=m^a_{22}$. For the states
		$\mathbf{|{\mathcal{\xi}}_{23}^{\pm}\rangle_{A}|1\rangle_{A}}$ we
		finally get $m^a_{11}=m^a_{22}=m^a_{33} .$ Therefore,
		${\mathcal{M}_{A}^m}^\dagger{\mathcal{M}_{A}^m}$
		$=\operatorname{diag}\left(\delta_{0}, \gamma, \gamma, \gamma\right) .$
		If possible let us assume that $\delta_{0}\neq0$ and
		$\gamma\neq0$. Then after Alice's measurement, Bob should do a
		nontrivial operation on his own system according to Alice's result. We
		denote $\mathcal{M}_{B}^m$ as Bob's operator. As we discussed above,
		by choosing suitable pair of states we can conclude that all the
		off-diagonal element of ${\mathcal{M}_{B}^m}^\dagger\mathcal{M}_{B}^m$
		is equal to 0. Similarly for the diagonal element as we have discussed
		above, if we take $\mathbf{
			|0\rangle_{A}|{\mathcal{X}}_{01}^{\pm}\rangle_{B},\;
			|0\rangle_{A}|{\mathcal{X}}_{23}^{\pm}\rangle_{B},\;|1\rangle_{A}|{\mathcal{X}}_{12}^{\pm}\rangle_{B}}$
		we finally get $m^b_{00}=m^b_{11}=m^b_{22}=m^b_{33}.$ Therefore
		${\mathcal{M}_{B}^m}^\dagger\mathcal{M}_{B}^m$ is propotional to the
		identity operator, i.e.,
		${\mathcal{M}_{B}^m}^\dagger\mathcal{M}_{B}^m=\lambda_{0}I$, which is
		trivial operator and this contradicts our assumption. So, either
		$\delta_{0}=0$ or $\gamma=0 $.
		Notice that this result also suggests that these states cannot
		be distinguished if Bob goes first. Now it is clear that if Alice goes
		first with a diagonal operator i.e., $\delta_{0}=\gamma=1$, then the
		above set of states cannot be distinguished. So, Alice has to do
		non-trivial measurement first and this only happens when any one of
		$\delta_{0}$, $\gamma$ not equal to zero. For that Alice only has two
		outcome measurement operators: $\quad
		{\mathcal{M}_{A}^1}^\dagger\mathcal{M}_{A}^1=\operatorname{diag}(1,0,0,0)$
		and ${\mathcal{M}_{A}^2}^\dagger\mathcal{M}_{A}^2$ =
		$I-{\mathcal{M}_{A}^1}^\dagger\mathcal{M}_{A}^1$
		$=\operatorname{diag}(0,1,1,1)$, see Fig.~\ref{fig44-1}. If `$1$'
		clicks, Bob is able to distinguish the left states by projecting onto
		$\mathbf{\left|0\pm1\right\rangle}$ and $\mathbf{\left|2\pm3\right\rangle}$. If `$2$'
		clicks, it isolates the remaining states. It is then Bob's turn to do
		measurement. Following the method we used above, we can similarly
		prove that Bob's measurement must be
		${\mathcal{M}_{B}^1}^\dagger\mathcal{M}_{B}^1=\operatorname{diag}(1,0,0,0)$
		and ${\mathcal{M}_{B}^2}^\dagger\mathcal{M}_{B}^2$
		$=\operatorname{diag}(0,1,1,1)$. The process will repeat a finite
		number of times and for each measurement outcomes for both parties the
		set $\mathcal{S}_1$ transforms only to a distinguishable set. This
		implies the fact that if the set is distinguishable (local) then for
		all possible nontrivial measurements, it is impossible to transform
		the set into an indistinguishable one. In other words, the set
		$\mathcal{S}_1$ is not activable through orthogonality-preserving
		LOCC. Hence we complete the proof.
	\end{proof}
	%\textcolor{red}{
	\par 	From the above a key point appears. The structure of the
	product states suggests, for local discrimination of these states
	local operations and two-way classical communication is necessary.
	Also notice that it is straightforward to generalize the structure
	given in Fig.~\ref{fig44}. We just have to keep adding additional
	layer of titles following the pattern. Furthermore, in qubit-qudit
	case such a class is quite obvious \cite{bandhyopadhyay201}. Clearly,
	the two-qudit construction given in this paper is nontrivial.
	\begin{figure}[h!]
		\centering
		\includegraphics[width=0.38\textwidth]{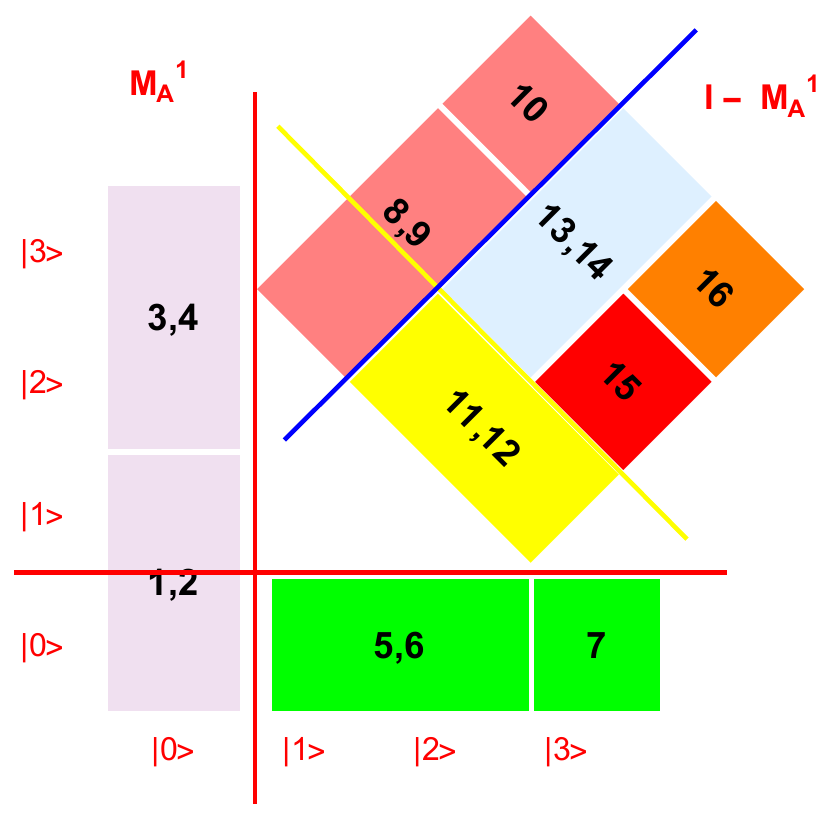}
		\caption{\emph{(Color online)} Representation of product states in
			${\mathbb{C}}^{4}\otimes{\mathbb{C}}^{4}$. Tile indices correspond to
			consecutively ordered basis states of set $\mathcal{S}_{1}$, while
			tile colors indicate compatible measurement setups for both parties.
			$\mbox{M}_i^j = {\mathcal{M}_i^j}^\dagger\mathcal{M}_i^j$; $i=\mbox{A,
				B}$, $j=1,2$ (this is also maintained in other figures unless
			explicitly stated).}\label{fig44-1}
	\end{figure}
	For bipartite systems the variation with respect to a hidden
	kind of non-locality is very limited, as $k$ can have only one value
	($k=1$). So, there are only two types of structures, one is activable
	and the other is non-activable. Both classes are weaker in the sense
	of non-locality because they are not locally indistinguishable class
	after all. Here we have provided only non-activable structure as it is 
    going to be useful further in the multipartite scenario. However, for 
    further research, one can also think about constructing activable scenarios 
    in two-ququad systems.
	
	%The tensor product structures always contradicts general intuition. Thats why the 
	Multipartite Hilbert spaces provide some interesting
	results which cannot be seen for bipartite Hilbert spaces. More
	generally, for the task of activation of non-locality the multipartite
	Hilbert space provides some broader view than bipartite cases.
	For example, in the tripartite scenario there exists a set of
	states $\mathcal{S}$, which is not activable when all three parties
	are spatially separated, i.e.,
	$\mathcal{H}^{\text{LOCC}}_{1}(\mathcal{S}) = 0$, but the same set of
	states might be activable when two parties perform some joint
	operation(s), i.e., $\mathcal{H}^{\text{LOCC}}_{2}(\mathcal{S})$ may
	not be zero \cite{subrata2024}. 	
	Consequently, a question arises: \emph{Is it feasible to construct a
		tripartite set for which the activation of non-locality by LOCC is
		precluded across every bipartition?} Such a class would, in essence,
	constitute the `worst case' from the perspective of non-locality
	activation. The subsequent findings furnish appropriate support for
	this aforementioned concept.
	Consider the set $\mathcal{S}_2
	\subseteq \mathbb{C}^4
	\otimes \mathbb{C}^2 \otimes \mathbb{C}^2$, given by, 
	
	\begin{multline}
		\mathcal{S}_2=
		\begin{Bmatrix}
			\mathbf{|\xi_0\rangle_{A}|{\mathcal{X}_0}\rangle_{B}|{\mathcal{Y}}_{01}^{\pm}\rangle_{C},}\;
			\mathbf{|\xi_0\rangle_{A}|{\mathcal{X}_1}\rangle_{B}|{\mathcal{Y}}_{01}^{\pm}\rangle_{C},}\;\\
			\mathbf{|{\mathcal{\xi}}_{12}^{\pm}\rangle_{A}|\mathcal{X}_0\rangle_{B}|\mathcal{Y}_0\rangle_{C},}\;
			%|2\rangle_{A}|{\mathcal{X}}_{23}^{\pm}\rangle_{B},\;
			\mathbf{|{\mathcal{\xi}}_{3}\rangle_{A}|\mathcal{X}_0\rangle_{B}|\mathcal{Y}_0\rangle_{C},}\;\\
			\mathbf{|\xi_1\rangle_{A}|{\mathcal{X}}_{1}\rangle_{B}|{\mathcal{Y}}_{0}\rangle_{C},}\;
			\mathbf{|{\mathcal{\xi}}_{3}\rangle_{A}|\mathcal{X}_1\rangle_{B}|\mathcal{Y}_0\rangle_{C}}\\
			\mathbf{|\xi_1\rangle_{A}|{\mathcal{X}}_{01}^{\pm}\rangle_{B}|\mathcal{Y}_1\rangle_{C},}\;
			\mathbf{|{\mathcal{\xi}}_{23}^{\pm}\rangle_{A}|\mathcal{X}_0\rangle_{B}|\mathcal{Y}_1\rangle_{C},}\;\\
			\mathbf{|{\mathcal{\xi}}_{3}\rangle_{A}|\mathcal{X}_1\rangle_{B}|\mathcal{Y}_1\rangle_{C},}\;	
		\end{Bmatrix}\label{2}
	\end{multline}
	where, $\mathbf{|{\mathcal{\xi}}_{ij}^{\pm}\rangle_{A}=\left(\frac{|i\rangle\pm |j\rangle}{\sqrt{2}}\right)_{A},\;|{\mathcal{X}}_{ij}^{\pm}\rangle_{B}=\left(\frac{|i\rangle\pm |j\rangle}{\sqrt{2}}\right)_{B}},\;$\\
	$
	\mathbf{|{\mathcal{Y}}_{ij}^{\pm}\rangle_{C}=\left(\frac{|i\rangle\pm |j\rangle}{\sqrt{2}}\right)_{C}}$ and $\mathbf{|{\mathcal{\xi}}_{k}\rangle_{A}=|k\rangle_{A},\;|{\mathcal{X}}_{k}\rangle_{B}=|k\rangle_{B}}$\\
	$\mathbf{|{\mathcal{Y}}_{k}\rangle_{C}=|k\rangle_{C}}$.

	\begin{prop}\label{prop2}
		
		The set $\mathcal{S}_2$ does not possess any activable non-locality in
		tripartition $A|B|C$ as well as in all bipartition under
		orthogonality-preserving LOCC, i.e.
		$(i)\; \mathcal{H}^{\text{LOCC}}_{1}(\mathcal{S}_2) = 0$ and $ (ii)\; \mathcal{H}^{\text{LOCC}}_{2}(\mathcal{S}_2) = 0.$
	\end{prop}
	\begin{proof} (i) We begin by noting that given any
		multipartite set, if it does not contain any activable non-locality
		across all bi-partitions then it becomes obvious that the set also does
		not contain any activable non-locality in multi-partitions. This is
		because in bi-partitions the operations are stronger than that of the
		multi-partitions. For example, in our context if we consider a
		bipartition then, two parties can perform joint measurements but in a
		tri-partition such a possibility is absent. Clearly, the operations in
		bi-partitions can be much stronger than that of a tri-partition. Thus,
		we concentrate only on proving the second part of the proposition. 
		
		(ii) To prove that $\mathcal{S}_2$ is not activable in
		all bi-partitions, consider the case A|BC. Here the states belong to a
		$\mathbb{C}^4 \otimes \mathbb{C}^4$ and they
		have the same forms as the states of the set $\mathcal{S}_1$. So by
		Proposition \ref{prop1}, the set of states $\mathcal{S}_2$ is non-activable in
		A|BC bi-partition.
		\par  For the bipartitions B|AC and C|AB, the states of the
		set $\mathcal{S}_2$ belong to $\mathbb{C}^2 \otimes
		\mathbb{C}^8$. Now, it is known that a set of product states in
		$\mathbb{C}^2 \otimes \mathbb{C}^n$ is always locally
		distinguishable \cite{BennettUPB1999}. Moreover, LOCC is not
		sufficient to create entanglement from product states. So, it is
		impossible to activate non-locality from the states of the set
		$\mathcal{S}_2$ in B|AC and C|AB bipartitions. Hence
		$\mathcal{H}^{\text{LOCC}}_{2}(\mathcal{S}_2) = 0$.
	\end{proof}
	\begin{rem} \label{rem3}
		Let us not concentrate on the particular structure of tripartite
		product states, given in (\ref{2}). Instead, we consider any
		tripartite orthogonal product states in $\mathbb{C}^4 \otimes
		\mathbb{C}^2 \otimes \mathbb{C}^2$ such that these states mimic the
		similar forms like the states of (\ref{1}) in $\mathbb{C}^4
		\otimes \mathbb{C}^4$ bipartition. Then, from the aforesaid proof
		technique it depicts that for such a set the activable non-locality is
		0 across all bi-partitions.
	\end{rem}
	\par 	In a three-qubit system, it is observed that all sets of
	orthogonal product states are non-activable across every bipartition.
	This deduction stems directly from the established fact that, in
	qubit-qudit scenarios, orthogonal product states consistently exhibit
	local distinguishability \cite{BennettUPB1999}. Consequently, from
	this standpoint, our proposed higher-dimensional construction presents
	a point of particular interest.

	\par 		Next, we want to discuss about a hierarchy among the
	multipartite locally distinguishable sets. For this purpose, we first
	consider the bipartite set $\mathcal{S}_3
	=\left\{\left|\phi_i\right\rangle_{A B}\right\}_{i=1}^{10}\in
	\mathbb{C}^6 \otimes \mathbb{C}^6$, where
	\begin{multline}
		%\mathcal{S}_3 = \left\{
		\begin{aligned}
			& \left|\phi_1\right\rangle_{A B}=|\mathbf{0}\rangle_A|\mathbf{0}-\mathbf{1}+\mathbf{4}-\mathbf{5}\rangle_B \\
			& \left|\phi_2\right\rangle_{A B}=|\mathbf{2}\rangle_A|\mathbf{1}-\mathbf{2}+\mathbf{5}-\mathbf{3}\rangle_B \\
			& \left|\phi_3\right\rangle_{A B}=|\mathbf{1-2}\rangle_A|\mathbf{0}-\mathbf{4}\rangle_B \\
			& \left|\phi_4\right\rangle_{A B}=|\mathbf{0-1}\rangle_A|\mathbf{2}-\mathbf{3}\rangle_B \\
			& \left|\phi_5\right\rangle_{A B}=|\mathbf{0+1+2}\rangle_A|\mathbf{0}+\mathbf{1}+\mathbf{2}+\mathbf{3}+\mathbf{4}+\mathbf{5}\rangle_B \\
			& \left|\phi_6\right\rangle_{A B}=|\mathbf{3}\rangle_A|\mathbf{0}-\mathbf{1}+\mathbf{4}-\mathbf{5}\rangle_B \\
			& \left|\phi_7\right\rangle_{A B}=|\mathbf{5}\rangle_A|\mathbf{1}-\mathbf{2}+\mathbf{5}-\mathbf{3}\rangle_B \\
			& \left|\phi_8\right\rangle_{A B}=|\mathbf{4-5}\rangle_A|\mathbf{0}-\mathbf{4}\rangle_B \\
			& \left|\phi_9\right\rangle_{A B}=|\mathbf{3-4}\rangle_A|\mathbf{2}-\mathbf{3}\rangle_B \\
			& \left|\phi_{10}\right\rangle_{A B}=|\mathbf{3+4+5}\rangle_A|\mathbf{0}+\mathbf{1}+\mathbf{2}+\mathbf{3}+\mathbf{4}+\mathbf{5}\rangle_B
		\end{aligned}
		%\right\}
		\label{5}
	\end{multline}
	It is quite straightforward to show that the set $\mathcal{S}_3$
	considered above is free from local redundancy
	\cite{bandhyopadhyay201,Li2022,subrata2024}. Here, Bob's system can be
	considered to be the composition of qubit and qutrit subsystems, 
	\begin{align*}
		|\mathbf{0}\rangle_B&:=|00\rangle_{b_1 b_2}, &&|\mathbf{1}\rangle_B:=|01\rangle_{b_1 b_2},\\
		|\mathbf{2}\rangle_B&:=|02\rangle_{b_1 b_2}, &&|\mathbf{3}\rangle_B:=|10\rangle_{b_1 b_2},\\
		\left|\mathbf{4}_B\right\rangle&:=|11\rangle_{b_1 b_2},&& |\mathbf{5}\rangle_B:=|12\rangle_{b_1 b_2}. 
	\end{align*}
	Take two states, $\left|\phi_3\right\rangle_{A B}$ and
	$\left|\phi_4\right\rangle_{A B}$. When any of the subparts (qubit or
	qutrit) of Bob's system for both states is discarded the reduced
	states will be nonorthogonal. Similar things happen for Alice also.
	This implies the set $\mathcal{S}_3$ is free from local redundancy.
	\par Now we will show that the set $\mathcal{S}_3$ is locally
	distinguishable. The players can avail the following discrimination
	protocol. First Bob performs a measurement: 	
	\begin{align*}
		\mathcal{M}_B \equiv& \{\mathcal{M}_B^1:=P\left[|\mathbf{0}-\mathbf{4}\rangle_B\right],  \mathcal{M}_B^2 :=P\left[|\mathbf{2}-\mathbf{3}\rangle_B\right], \\
		& \mathcal{M}_B^3:=P\left[\ket{ \mathbf{0}+\mathbf{1}+\mathbf{2}+\mathbf{3}+\mathbf{4} +\mathbf{5}}_B\right],\\
		& \mathcal{M}_B^4:=\mathbb{I}-\left(\mathcal{M}_B^1+\mathcal{M}_B^2+\mathcal{M}_B^3\right) \}.
	\end{align*} 
	\noindent	Here, $P\left[|\cdot\rangle\right]:=$
	$|\cdot\rangle\langle\cdot|_{\mathcal{P}}$, and $\mathcal{P}$ denotes
	the party. When $\mathcal{M}_B^1$ clicks, the given state must be
	$\left|\phi_3\right\rangle$ and $\left|\phi_8\right\rangle$, which can
	be distinguished by Alice, projecting onto $|\mathbf{1-2}\rangle$ and
	$|\mathbf{4-5}\rangle$. Similarly, for the click $\mathcal{M}_B^2$,
	the states are $\left|\phi_4\right\rangle$ and
	$\left|\phi_9\right\rangle$, which can be distinguished by Alice,
	projecting onto $|\mathbf{0-1}\rangle$ and $|\mathbf{3-4}\rangle$.
	Also for the outcome $\mathcal{M}_B^3$ the isolated states are
	$\left|\phi_5\right\rangle$ and $\left|\phi_{10}\right\rangle$, which
	can be distinguished by Alice, projecting onto
	$|\mathbf{0+1+2}\rangle$ and $|\mathbf{3+4+5}\rangle$. Whenever
	$\mathcal{M}_B^4$ clicks the given state can be
	$\left|\phi_1\right\rangle$, $\left|\phi_2\right\rangle$,
	$\left|\phi_6\right\rangle$ and $\left|\phi_7\right\rangle$. However,
	in that case, Alice can perform a measurement 
	\begin{align*}
		\mathcal{M}_A &\equiv\left\{\mathcal{M}_A^1:=\right. P\left[|\mathbf{0}\rangle_A\right], \mathcal{M}_A^2 :=P\left[|\mathbf{2}\rangle_A\right],\\
		& \mathcal{M}_A^3:=P[\ket{\mathbf{3}}_A], 
		\mathcal{M}_A^4:=\mathbb{I}-\left(\mathcal{M}_A^1+\mathcal{M}_A^2+\mathcal{M}_A^3\right)\},
	\end{align*}
	to distinguish between these four states. This concludes the
	local discrimination protocol for the set $\mathcal{S}_3$. In the
	following, we will demonstrate a protocol to activate non-locality
	without entanglement from this set.
	\begin{prop}
		The set $\mathcal{S}_3$ is a locally distinguishable set and
		can be transformed deterministically to a locally irreducible set via
		orthogonality-preserving LOCC.    
	\end{prop}
	\begin{proof} 
		Consider that Bob performs a local measurement 
		\begin{align*}
			\mathcal{K}_B &\equiv \left\{\mathcal{K}^B_1:=P\left[(|\mathbf{0}\rangle,|\mathbf{1}\rangle,|\mathbf{2}\rangle)_B\right], \mathcal{K}^B_2:=P\left[(|\mathbf{3}\rangle,|\mathbf{4}\rangle, |\mathbf{5}\rangle)_B\right]\right\},\\
			&P\left[(|i\rangle,|j\rangle,\dots)_{\mathcal{P}}\right] = \left[(|i\rangle\langle i|+|j\rangle\langle j|+\dots)_{\mathcal{P}}\right],
		\end{align*}
		$\mathcal{P}$ stands for party. If $\mathcal{K}^B_1$ clicks, they end up with,
		\[
		\left\{\begin{array}{c}
			|\mathbf{0}\rangle_A|\mathbf{0}-\mathbf{1}\rangle_B,|\mathbf{2}\rangle_A|\mathbf{1}-\mathbf{2}\rangle_B, \\
			|\mathbf{1-2}\rangle_A|\mathbf{0}\rangle_B,|\mathbf{0-1}\rangle_A|\mathbf{2}\rangle_B, \\
			|\mathbf{0+1+2}\rangle_A|\mathbf{0}+\mathbf{1}+\mathbf{2}\rangle_B,\\
			|\mathbf{3}\rangle_A|\mathbf{0}-\mathbf{1}\rangle_B,|\mathbf{5}\rangle_A|\mathbf{1}-\mathbf{2}\rangle_B, \\
			|\mathbf{4-5}\rangle_A|\mathbf{0}\rangle_B,|\mathbf{3-4}\rangle_A|\mathbf{2}\rangle_B, \\
			|\mathbf{3+4+5}\rangle_A|\mathbf{0}+\mathbf{1}+\mathbf{2}\rangle_B
		\end{array}\right\}
		\]
		After that Alice makes measurement  $\mathcal{K}_A \equiv\left\{\mathcal{K}^A_1:=P\left[(|\mathbf{0}\rangle,|\mathbf{1}\rangle,|\mathbf{2}\rangle)_A\right], \mathcal{K}^A_2:=\right.$ $\left.P\left[(|\mathbf{3}\rangle,|\mathbf{4}\rangle,|\mathbf{5}\rangle)_A\right]\right\}$. If $\mathcal{K}^A_1$ occurs, it gives,
		\[
		\left\{\begin{array}{c}
			|\mathbf{0}\rangle_A|\mathbf{0}-\mathbf{1}\rangle_B,|\mathbf{2}\rangle_A|\mathbf{1}-\mathbf{2}\rangle_B, \\
			|\mathbf{1-2}\rangle_A|\mathbf{0}\rangle_B,|\mathbf{0-1}\rangle_A|\mathbf{2}\rangle_B, \\
			|\mathbf{0+1+2}\rangle_A|\mathbf{0}+\mathbf{1}+\mathbf{2}\rangle_B
		\end{array}\right\}
		\]
		which is a locally irreducible set \cite{BennettUPB1999}. If $\mathcal{K}^A_2$ occurs, it also gives a locally irreducible set,
		\[ 
		\left\{\begin{array}{c}
			|\mathbf{3}\rangle_A|\mathbf{0}-\mathbf{1}\rangle_B,|\mathbf{5}\rangle_A|\mathbf{1}-\mathbf{2}\rangle_B, \\
			|\mathbf{4-5}\rangle_A|\mathbf{0}\rangle_B,|\mathbf{3-4}\rangle_A|\mathbf{2}\rangle_B, \\
			|\mathbf{3+4+5}\rangle_A|\mathbf{0}+\mathbf{1}+\mathbf{2}\rangle_B
		\end{array}\right\}.
		\]
		On the other hand, if Bob gets $\mathcal{K}^B_2$, they are then left with the following states,
		\[	\left\{\begin{array}{c}
			|\mathbf{0}\rangle_A|\mathbf{4}-\mathbf{5}\rangle_B,|\mathbf{2}\rangle_A|\mathbf{5}-\mathbf{3}\rangle_B, \\
			|\mathbf{1-2}\rangle_A|\mathbf{4}\rangle_B,|\mathbf{0-1}\rangle_A|\mathbf{3}\rangle_B, \\
			|\mathbf{0+1+2}\rangle_A|\mathbf{3}+\mathbf{4}+\mathbf{5}\rangle_B,\\
			|\mathbf{3}\rangle_A|\mathbf{4}-\mathbf{5}\rangle_B,|\mathbf{5}\rangle_A|\mathbf{5}-\mathbf{3}\rangle_B, \\
			|\mathbf{4-5}\rangle_A|\mathbf{4}\rangle_B,|\mathbf{3-4}\rangle_A|\mathbf{3}\rangle_B, \\
			|\mathbf{3+4+5}\rangle_A|\mathbf{3}+\mathbf{4}+\mathbf{5}\rangle_B
		\end{array}\right\}.
		\]
		After that Alice makes a measurement  
		\begin{align*}
			\mathcal{K}_A \equiv&\left\{\mathcal{K}^A_1:=P\left[(|\mathbf{0}\rangle,|\mathbf{1}\rangle,|\mathbf{2}\rangle)_A\right] \right., \\
			& \left. \mathcal{K}^A_2:=P\left[(|\mathbf{3}\rangle,|\mathbf{4}\rangle,|\mathbf{5}\rangle)_A\right]\right\}.
		\end{align*}
		If $\mathcal{K}^A_1$ occurs, it gives,
		\[ 
		\left\{\begin{array}{c}
			|\mathbf{0}\rangle_A|\mathbf{4}-\mathbf{5}\rangle_B,|\mathbf{2}\rangle_A|\mathbf{5}-\mathbf{3}\rangle_B, \\
			|\mathbf{1-2}\rangle_A|\mathbf{4}\rangle_B,|\mathbf{0-1}\rangle_A|\mathbf{3}\rangle_B, \\
			|\mathbf{0+1+2}\rangle_A|\mathbf{3}+\mathbf{4}+\mathbf{5}\rangle_B
		\end{array}\right\},
		\]
		which is a locally irreducible set. If $\mathcal{K}^A_2$ occurs, it also gives a locally irreducible set,
		\[ 
		\left\{\begin{array}{c}
			|\mathbf{3}\rangle_A|\mathbf{4}-\mathbf{5}\rangle_B,|\mathbf{5}\rangle_A|\mathbf{5}-\mathbf{3}\rangle_B, \\
			|\mathbf{4-5}\rangle_A|\mathbf{4}\rangle_B,|\mathbf{3-4}\rangle_A|\mathbf{3}\rangle_B, \\
			|\mathbf{3+4+5}\rangle_A|\mathbf{3}+\mathbf{4}+\mathbf{5}\rangle_B
		\end{array}\right\}.
		\]
		It is evident that, for each instance of Alice's measurement,
		specified by the set  $\{\mathcal{K}^A_1,\;\mathcal{K}^A_2\}$; the
		five post-measurement states, contingent upon  $\mathcal{K}^B_1$
		clicking, yields the celebrated unextendable product basis (UPB)
		\cite{BennettPB1999,BennettUPB1999} in $\mathbb{C}^3 \otimes \mathbb{C}^3$.
		Also, the post-measurement states for each case of Alice's measurement
		$\mathcal{K}^A_1,\;\mathcal{K}^A_2$ when $\mathcal{K}^B_2$ clicks form
		the same UPB. See Fig.~\ref{f44-3}. It has been well established that
		UPB is locally indistinguishable \cite{Divincinzo,BennettUPB1999}. So,
		the set $\mathcal{S}_3$ is activable by orthogonality-preserving LOCC.
		Hence, this completes the proof.
	\end{proof}
	\begin{figure}[h!]
		\centering
		\includegraphics[width=0.48\textwidth]{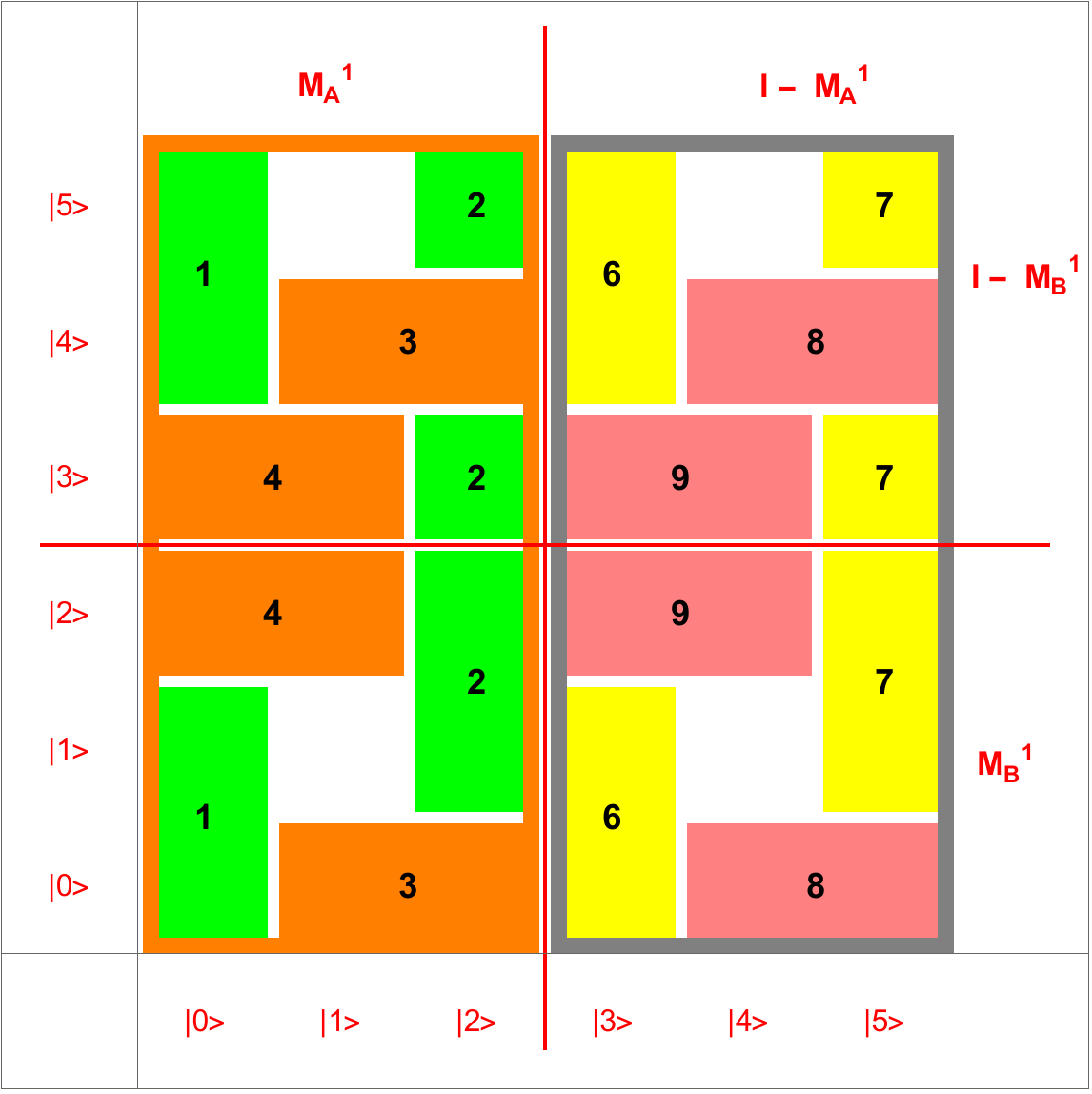}
		\caption{\emph{(Color online)} Tiling diagram for the states in \( \mathbf{\mathcal{S}_3} \). The outlined region indicates the support of Alice's and Bob's measurement outcomes, %\( \mathcal{K}_1^B \), 
			resulting in post-measurement states, %are all 
			contained in a UPB subspace. $\mbox{M}_i^j = {\mathcal{K}_j^i}$; $i=\mbox{A, B}$, $j=1,2$}\label{f44-3}
	\end{figure}

	{\it Towards a hierarchy}: We consider the set $\mathcal{S}_4 =
	\left\{\left|\phi_i\right\rangle_{AB} \otimes \left|\mathbf{0}\right\rangle_C,
	\left|\phi_i\right\rangle_{AB} \otimes \left|\mathbf{1}\right\rangle_C
	\right\}_{i=1}^{10} \in \mathbb{C}^6 \otimes \mathbb{C}^6 \otimes
	\mathbb{C}^2$, where
	$\left\{\left|\phi_i\right\rangle_{AB}\right\}_{i=1}^{10} =
	\mathcal{S}_3$. 
	Now, consider all bipartitions of the tripartite system. For the
	bipartition $AB|C$, the total Hilbert space is $\mathbb{C}^{36}\otimes
	\mathbb{C}^{2}$, and due to the limited dimension of subsystem $C$,
	the set remains non-activable in this cut \cite{bandhyopadhyay201}.
	However, for the bipartitions $A|BC$ and $B|AC$, the set becomes
	activable. This follows directly from Proposition 3. 
	
	\begin{rem} 
		Let us now highlight the contrasting structures of the sets
		$\mathcal{S}_2$ and $\mathcal{S}_4$.
		The set $\mathcal{S}_2$ is a tripartite ensemble of orthogonal quantum
		states that is initially locally distinguishable and remains
		non-activable in all bipartitions. In contrast, the set
		$\mathcal{S}_4$ 
		is activable in certain bipartitions (but not in every bipartition).
		This structural difference reveals a clear separation in the degrees
		of hidden non-locality for $\mathcal{S}_2$ and $\mathcal{S}_4$. 
	\end{rem}

	So far, we have discussed about the product states only. Nevertheless,
	in the following, we include entangled states into our discussions. 
	
	\section{Non-activable entangled states}\label{A3}
	Here we consider several sets that are local and non-activable by
	LOCC, i.e., $\mathcal{H}^{\text{LOCC}}_{1}(\mathcal{S}_j) = 0$, for
	different sets $\mathcal{S}_j$. 
	$\mathbf{|^\theta\mathcal{W}_{ij,kl}^{\pm}\rangle_{AB}=|\theta\rangle_{A}|{\mathcal{X}}_{ij}^{\pm}\rangle_{B}+|\bar{(\theta+2)}\rangle_{A}|{\mathcal{X}}_{kl}^{\pm}\rangle_{B}},\;$
	and
	$\mathbf{|^\theta\mathcal{W}_{ij,m}^{\pm}\rangle_{AB}=|\theta\rangle_{A}|{\mathcal{X}}_{ij}^{\pm}\rangle_{B}+|\bar{(\theta+2)}\rangle_{A}|{\mathcal{X}}_{m}\rangle_{B}},\;$
	also 
	$\mathbf{|^\theta\mathcal{\bar{W}}_{ij,kl}^{\pm}\rangle_{AB}=|{\mathcal{\xi}}_{ij}^{\pm}\rangle_{A}|\theta\rangle_{B}+|{\mathcal{\xi}}_{kl}^{\pm}\rangle_{A}|\bar{(\theta+2)}\rangle_{B}},\;$
	and
	$\mathbf{|^\theta\mathcal{\bar{W}}_{ij,m}^{\pm}\rangle_{AB}=|{\mathcal{\xi}}_{ij}^{\pm}\rangle_{A}|\theta\rangle_{B}+|{\mathcal{\xi}}_{m}\rangle_{A}|\bar{(\theta+2)}\rangle_{B}},\;$,
	for $\mathbf{\theta=0,1}$ and $\mathbf{\bar{(\theta+2)}}$ denotes $\mathbf{(\theta+2)}$ modulo $\mathbf{d}$.
	Here we consider the set $\mathcal{S}_5
	=\left\{\left|\phi_i\right\rangle_{A B}\right\}\in
	\mathbb{C}^4 \otimes \mathbb{C}^4$, which contains product
	states as well as entangled states. The set is given by-
	\begin{equation}\label{3}
		\mathcal{S}_5=
		\begin{Bmatrix}
			\mathbf{|^0\mathcal{W}_{01,23}^{\pm}\rangle_{AB},}\;
			\mathbf{|0\rangle_{A}|{\mathcal{X}}_{23}^{\pm}\rangle_{B},}\;
			\mathbf{|^0\mathcal{\bar{W}}_{12,3}^{+}\rangle_{AB},}\;\\
			\mathbf{|{\mathcal{\xi}}_{12}^{-}\rangle_{A}|0\rangle_{B},}\;
			\mathbf{|{\mathcal{\xi}}_{3}\rangle_{A}|0\rangle_{B},}\;
			\mathbf{|^1\mathcal{W}_{12,3}^{+}\rangle_{AB},}\;\\
			\mathbf{|1\rangle_{A}|{\mathcal{X}}_{12}^{-}\rangle_{B},}\;
			\mathbf{|1\rangle_{A}|{\mathcal{X}}_{3}\rangle_{B},}\;
			\mathbf{|{\mathcal{\xi}}_{23}^{\pm}\rangle_{A}|1\rangle_{B}}
		\end{Bmatrix}
	\end{equation}
	
	\begin{prop}
		The set $\mathcal{S}_5$ does not possess any activable non-locality under orthogonality-preserving LOCC. That is, its hidden non-locality $\mathcal{H}^{\text{LOCC}}_{1}(\mathcal{S}_5) = 0.$    
	\end{prop}
	\begin{proof} The states in (\ref{3}) can be viewed as superpositions of the product states given in (\ref{1}), with the key difference being the inclusion of entangled states. Hence, the overall structure of the proof follows similarly to that of Proposition \ref{prop1}, with the key steps sketched below.
		\par     Without loss of generality, let us assume that Alice goes first. By analyzing suitable pairs of states, such as $\mathbf{|0\rangle_{A}|{\mathcal{X}}_{23}^{+}\rangle_{B}}$ and $\mathbf{|1\rangle_{A}|{\mathcal{X}}_{3}\rangle_{B}}$, and using the orthogonality-preserving condition, we find that the off-diagonal terms of Alice's measurement operator must vanish. Repeating this for other carefully chosen state pairs, we conclude that Alice's operator must be diagonal. Furthermore, symmetry in the state structure leads to equal diagonal entries: $\delta_1 = \delta_2 = \delta_3 = \gamma$. Therefore, we obtain:
		\[
		{\mathcal{M}_{A}^m}^{\dagger}\mathcal{M}_{A}^m = \operatorname{diag}(\delta_0, \gamma, \gamma, \gamma).
		\]
		
		Assuming both $\delta_0 \neq 0$ and $\gamma \neq 0$, Bob must apply a nontrivial operation. However, a similar analysis on Bob's side implies that his operator must also be diagonal and further turns out to be proportional to the identity: $
		{\mathcal{M}_{B}^m}^{\dagger}\mathcal{M}_{B}^m = \lambda_0 I,
		$
		which contradicts the requirement of a nontrivial measurement if the set is to be initially locally distinguishable. Hence, either $\delta_0 = 0$ or $\gamma = 0$, implying that Alice must begin with a nontrivial two-outcome measurement:
		\[
		{\mathcal{M}_{A}^{1}}^\dagger \mathcal{M}_{A}^{1} = \operatorname{diag}(1, 0, 0, 0), \quad
		{\mathcal{M}_{A}^{2}}^\dagger \mathcal{M}_{A}^{2} = \operatorname{diag}(0, 1, 1, 1).
		\]
		If the outcome `$1$' clicks, Bob can distinguish the remaining states by projections onto $\mathbf{\left|0\pm1\right\rangle}$ and $\mathbf{\left|2\pm3\right\rangle}$. If the outcome is `$2$', the remaining subspace can again be handled by Bob using a similar measurement structure. See Fig.~\ref{fig44-4} for the measurement-induced partitioning of the state space.
		
		\begin{figure}[h]
			\centering
			\includegraphics[width=0.4\textwidth]{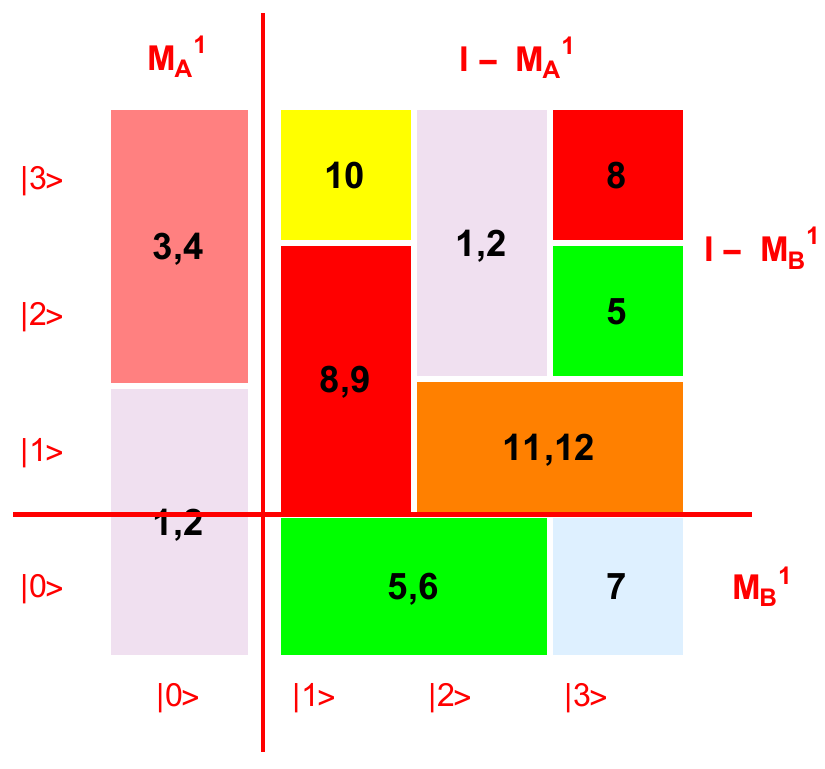}
			\caption{\emph{(Color online)} Tile structure of the states in ${\mathbb{C}}^{4}\otimes{\mathbb{C}}^{4}$, given in (\ref{3}). The indices of the tiles follow the ordering of the states in $\mathcal{S}_5$. Tile colors represent the measurement configuration for both parties.}
			\label{fig44-4}
		\end{figure}
		
		The process continues for a finite number of steps. At each step, the remaining set remains locally distinguishable. Therefore, the set $\mathcal{S}_5$ cannot be transformed into an indistinguishable one through any orthogonality-preserving LOCC sequence. This shows that $\mathcal{S}_5$ is \textit{not activable} under LOCC, and this completes the proof.
	\end{proof}
	\begin{figure}[!hb]
		\centering
		\includegraphics[width=0.40\textwidth]{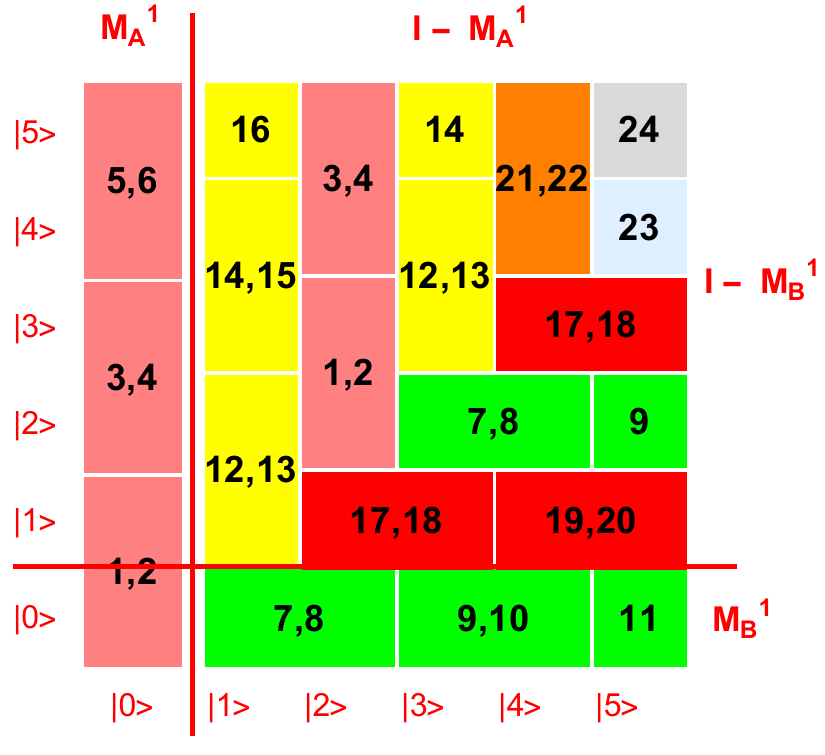}
		\caption{\emph{(Color online)} Tiles representation of states in
			$\mathbb{C}^6 \otimes \mathbb{C}^6$, indexed by states of
			$\mathcal{S}_6$ in order and colored according to the possibility of
			simultaneous local measurements by both parties.
		}
		\label{fig44-5}
	\end{figure}
	It is not very difficult to construct the set $\mathcal{S}_5$
	in arbitrary higher dimensions from its hereditary symmetry. For the
	case of higher dimensions, the only change will be the number of
	classical rounds required for discrimination task. Eventually for
	higher dimensions, the LOCC round numbers drastically increase for the
	corresponding tasks, but for each round the post-measurement states
	becomes distinguishable (local). Also, one can find the trade off
	between the dimensions of the systems and the corresponding required
	LOCC round number for the discrimination tasks. Now, consider the set
	$\mathcal{S}_6 =\left\{\left|\phi_i\right\rangle_{A
		B}\right\}\in\mathbf{\mathbb{C}^6 \otimes \mathbb{C}^6}$, by
	\begin{multline}
		\mathcal{S}_6=
		\begin{Bmatrix}
			\mathbf{|^0\mathcal{W}_{01,23}^{\pm}\rangle_{AB},}\;
			\mathbf{|^0\mathcal{W}_{23,45}^{\pm}\rangle_{AB},}\;
			\mathbf{|0\rangle_{A}|{\mathcal{X}}_{45}^{\pm}\rangle_{B},}\;\\
			\mathbf{|^0\mathcal{\bar{W}}_{12,34}^{\pm}\rangle_{AB},}\;
			\mathbf{|^0\mathcal{\bar{W}}_{34,5}^{+}\rangle_{AB},}\;
			\mathbf{|{\mathcal{\xi}}_{34}^{-}\rangle_{A}|0\rangle_{B},}\;\\
			\mathbf{|{\mathcal{\xi}}_{5}\rangle_{A}|0\rangle_{A},}\;
			\mathbf{|^1\mathcal{W}_{12,34}^{\pm}\rangle_{AB},}\;
			\mathbf{|^1\mathcal{W}_{34,5}^{+}\rangle_{AB},}\;\\
			\mathbf{|1\rangle_{A}|{\mathcal{X}}_{34}^{-}\rangle_{B},}\;
			\mathbf{|1\rangle_{A}|{\mathcal{X}}_{5}\rangle_{B},}\;
			\mathbf{|^1\mathcal{\bar{W}}_{23,45}^{\pm}\rangle_{AB},}\;\\
			\mathbf{|{\mathcal{\xi}}_{45}^{\pm}\rangle_{A}|0\rangle_{B},}
			\mathbf{|4\rangle_{A}|{\mathcal{X}}_{45}^{\pm}\rangle_{B},\;}
			\mathbf{|{\mathcal{\xi}}_{5}\rangle_{A}|4\rangle_{A},}\;\\
			\mathbf{|\xi_{5}\rangle_{A}|{\mathcal{X}}_{5}\rangle_{B}}
		\end{Bmatrix}\label{4}
	\end{multline}
	By the similar technique as given for (\ref{3}), it is
	possible to show that the set $\mathcal{S}_6$ does not possess any
	activable non-locality under orthogonality-preserving LOCC, i.e.,
	$\mathcal{H}^{\text{LOCC}}_{1}(\mathcal{S}_6) = 0$. See
	Fig.~\ref{fig44-5}.	
	Here we now discuss the following hierarchy. 
	The sets of states considered in (\ref{4}) and (\ref{5}) are
	equally local when considered with respect to perfect discrimination
	by LOCC, as in both cases, the sets are perfectly distinguishable by
	LOCC. But the consideration of hidden non-locality provides us the
	privilege to put a hierarchy among the sets. Precisely, we can claim
	that the sets of (\ref{4}) are more local compared to those of
	(\ref{5}), because the latter class contains hidden non-locality while
	for the former case there is no hidden non-locality though 
	the set of (\ref{4}) contains entangled states but the set of (\ref{5}) does not. 
	\section{Conclusion}\label{A4}

	In this manuscript, we have presented structures for locally
	distinguishable product states and entangled states such that they
	cannot be transformed to a locally indistinguishable set under
	orthogonality-preserving LOCC. Furthermore, we have constructed sets
	of multipartite states which are not activable in any bipartition. In
	other words, such locally distinguishable sets cannot be transformed
	to a locally indistinguishable set in any bipartition under
	orthogonality-preserving LOCC. This is, in fact, the worst case
	scenario in view of non-locality activation. We also have constructed
	locally distinguishable sets which can be transformed to locally
	indistinguishable sets under orthogonality-preserving LOCC. Then, we
	have classified the locally distinguishable states by introducing
	hierarchies. 
	The structures that we have provided here can be easily generalized for high-dimensional Hilbert spaces. 
	Finally, we have compared between locally distinguishable product states and entangled states.
	 
	\section*{ACKNOWLEDGMENT}
	SH acknowledges partial funding by the European Union under Horizon
	Europe (grant agreement no.~101080086). Views and opinions expressed
	are however those of the author(s) only and do not necessarily reflect
	those of the European Union or the European Commission. Neither the
	European Union nor the granting authority can be held responsible for
	them.

\end{document}